\pdfoutput=1
\RequirePackage{ifpdf}
\ifpdf % We~are running pdfTeX in pdf mode
\documentclass[pdftex]{sigma}%,draft
\else
\documentclass{sigma}
\fi

\numberwithin{equation}{section}

\newtheorem{Theorem}{Theorem}[section]
\newtheorem{Corollary}[Theorem]{Corollary}
 { \theoremstyle{definition}
\newtheorem{Definition}[Theorem]{Definition}
\newtheorem{Remark}[Theorem]{Remark} }

\begin{document}
\allowdisplaybreaks

\newcommand{\arXivNumber}{2102.04207}

\renewcommand{\PaperNumber}{088}

\FirstPageHeading

\ShortArticleName{Lax Pair for a Novel Two-Dimensional Lattice}

\ArticleName{Lax Pair for a Novel Two-Dimensional Lattice}

\Author{Maria N.~KUZNETSOVA}

\AuthorNameForHeading{M.N.~Kuznetsova}

\Address{Institute of Mathematics, Ufa Federal Research Centre, Russian Academy of Sciences,\\
112 Chernyshevsky Street, Ufa 450008, Russia}
\Email{\href{mailto:mariya.n.kuznetsova@gmail.com}{mariya.n.kuznetsova@gmail.com}}

\ArticleDates{Received February 09, 2021, in final form September 15, 2021; Published online September 26, 2021}

\Abstract{In paper by I.T.~Habibullin and our joint paper the algorithm for classification of integrable equations with three independent variables was proposed. This method is based on the requirement of the existence of an infinite set of Darboux integrable reductions and on the notion of the characteristic Lie--Rinehart algebras. The method was applied for the classification of integrable cases of different subclasses of equations $u_{n,xy} = f(u_{n+1},u_n,u_{n-1}, u_{n,x},u_{n,y})$ of special forms. Under this approach the novel integrable chain was obtained. In present paper we construct Lax pair for the novel chain. To construct the Lax pair, we use the scheme suggested in papers by E.V.~Ferapontov. We also study the periodic reduction of the chain.}

\Keywords{Lax pair; two-dimensional lattice; integrable reduction; characteristic algebra; Lie--Rinehart algebra; Darboux integrable system; higher symmetry; $x$-integral}

\Classification{37K10; 37K30; 37D99}

\section{Introduction}

In a number of recent publications \cite{FHKN,H2013, HP1, HabKuzS, HP2, Kuzn} the problem of integrable classification of two-dimensional lattices
\begin{equation} \label{eq0}
u_{n,xy} = f(u_{n+1},u_n,u_{n-1}, u_{n,x},u_{n,y}), \qquad -\infty < n < \infty,
\end{equation}
was studied. Here the sought function $u_n = u_n(x,y)$ depends on the real variables $x$, $y$ and the integer variable $n$.
In these papers we proposed the method for seeking and classifying integrable equations with three independent variables based on the requirement of the existence of a set of Darboux integrable reductions and on the notion of the characteristic Lie--Rinehart algebras. The method was applied to different subclasses of equations (\ref{eq0}) of special forms.

Within this approach we use the following
\begin{Definition} \label{definition1}
A lattice of the form \eqref{eq0}
is called integrable if there exist locally analytic functions $\varphi$ and $\psi$ of two variables such that for any choice of integers $N_1$, $N_2$ the hyperbolic type system
\begin{gather}
u_{N_1,xy} = \varphi(u_{N_1+1},u_{N_1}), \nonumber \\
u_{n,xy}=f(u_{n+1}, u_{n}, u_{n-1}, u_{n,x}, u_{n,y}),\qquad N_1 < n < N_2, \label{finite_sys} \\
u_{N_2,xy}=\psi(u_{N_2},u_{N_2-1}), \nonumber
\end{gather}
obtained from lattice (\ref{eq0}) by imposing cut-off conditions at $n=N_1$ and $n=N_2$, is integrable in the sense of Darboux.
\end{Definition}
Let us recall what Darboux integrability means.
\begin{Definition}
 A function $I=I(x,\bar u,\bar u_x,\bar u_{xx},\dots)$ is called an $y$-integral if it satisfies the equation $D_yI=0$ for every solution of system \eqref{finite_sys}. A function $J=J(y,\bar u,\bar u_y,\bar u_{yy},\dots)$ is called a $x$-integral if it satisfies the equation $D_x J=0$. Integrals of the form $I=I(x)$ and $J=J(y)$ are called trivial.
\end{Definition}
Here $\bar u$ is a vector $\bar u=(u_{N_1}, u_{N_1 + 1}, \dots, u_{N_2})$, $\bar u_x$ is its derivative and so on. The operators $D_y$ and $D_x$ are operators of the total derivative with respect to the variable $y$ or $x$, correspondingly, by virtue of system~(\ref{finite_sys}).
\begin{Definition}
A system \eqref{finite_sys} is called Darboux integrable if it possesses $N_2 - N_1 + 1$ functionally independent nontrivial integrals in both characteristic directions $x$ and $y$.
\end{Definition}
Darboux integrable systems are amenable to study by the Lie--Rinehart algebras. Let $I=I(x,\bar u,\bar u_x,\bar u_{xx},\dots )$ be a nontrivial $y$-integral for the system (\ref{finite_sys}). Then $I$ must satisfy the following system:
\begin{equation*}
YI=0, \qquad X_i I = 0,
\end{equation*}
where
\begin{equation*}\label{Ydef}
X_i = \frac{\partial}{\partial u_{i,y}},\qquad Y=\sum_{i=N_1}^{N_2} \left(u_{i,y} \frac{\partial}{\partial u_i} + f_i \frac{\partial}{\partial u_{i,x}} + D_x(f_{i})\frac{\partial}{\partial u_{i,xx}} + \cdots \right)
\end{equation*}
and $f_i=f(u_{i+1},u_i,u_{i-1},u_{i,x},u_{i,y})$. The first equation follows from the fact that the operator $D_y$ acts on functions $I=I(x,\bar u,\bar u_x,\bar u_{xx},\dots)$ by the rule $D_y I = Y I$, the second one arises because~$I$ doesn't depend on variables $u_{i,y}$.

Let us consider the Lie algebra $L_y$ generated by the operators $Y$, $X_i$ over the ring $A$ of locally analytic functions of the dynamical variables $\bar u_{y},\bar u,\bar u_x,\bar u_{xx},\dots$. To the standard multiplication operation $[Z, W] = ZW - WZ$ we add two conditions: $[Z,aW]=Z(a)W+a[Z,W]$ and $(aZ)b=aZ(b)$ valid for any $Z,W\in L_y$ and $a,b\in A$. These equalities means that for any $Z \in L_y$ and any $a \in A$, the element $a Z \in L_y$. In this case the algebra $L_y$ is called the Lie--Rinehart algebra~\cite{Million, Rinehart}.

If there exists a finite basis $Z_1,Z_2,\dots,Z_k\in L_y$ such that an arbitrary element $Z\in L_y$ is represented as a linear combination
$Z=a_1Z_1+a_2Z_2+\dots +a_kZ_k$, where coefficients $a_1,a_2,\dots,a_k\in A$; and if the equality $Z = 0$ implies that $a_1 = a_2 = \cdots = a_k = 0$, then algebra $L_y$ is of a finite dimension.

The integrability criterion of the hyperbolic type system in the sense of Darboux is formulated as follows \cite{ZhiberK, ZMHSbook}:
\begin{Theorem} \label{Theorem 1}
System \eqref{finite_sys} admits a complete set of the $y$-integrals $($a complete set of the $x$-integrals$)$ if and only if its characteristic algebra $L_y$ $($respectively, characteristic algebra $L_x$$)$ is of finite dimension.
\end{Theorem}

\begin{Corollary} \label{corollary1}
System \eqref{finite_sys} is integrable in the sense of Darboux if both characteristic algebras~$L_x$ and~$L_y$ are of finite dimension.
\end{Corollary}

The above statements play a key role in our classification works. Within the scope of this paper we need one of our results:
paper \cite{HP2} provides a complete list of integrable two-dimensional lattices of the form
\begin{gather}
u_{n,xy} = \alpha(u_{n+1}, u_n, u_{n-1})u_{n,x}u_{n,y} + \beta(u_{n+1}, u_n, u_{n-1}) u_{n,x}\nonumber \\
\hphantom{u_{n,xy} =}{} +\gamma(u_{n+1}, u_n, u_{n-1}) u_{n,y} + \delta(u_{n+1}, u_n, u_{n-1}),\label{quasi}
\end{gather}
with the coefficient $\alpha$ satisfying the conditions $\frac{\partial \alpha(u_{n+1}, u_n, u_{n-1})}{\partial u_{n \pm 1}} \neq 0$. This list consists of two equations:

\begin{Theorem} \label{T1}
Integrable equation of the form \eqref{quasi} can be reduced by a point transformation to one of the following forms:
\begin{gather}
 u_{n,xy}=\alpha_nu_{n,x}u_{n,y}, \label{eq1}\\
 u_{n,xy} = \alpha_n\big(u_{n,x} - u^2_n - 1\big)\big(u_{n,y} - u^2_n - 1\big) + 2 u_n\big(u_{n,x}+u_{n,y}-u^2_n - 1\big), \label{eq3}
\end{gather}
where
\begin{equation*}
\alpha_n = \frac{1}{u_n - u_{n-1}} - \frac{1}{u_{n+1}-u_n}=\frac{u_{n+1} - 2 u_n + u_{n-1}}{(u_{n+1}-u_n)(u_n - u_{n-1})}.
\end{equation*}
\end{Theorem}

Equation \eqref{eq1} was found before in papers \cite{Fer-TMF, ShY} by Ferapontov and Shabat and Yamilov.
Equation \eqref{eq3} appeared in \cite{HP2} as a result of the classification procedure.

The aim of the paper is to find Lax pair for novel chain \eqref{eq3}, to explain the method of finding Lax pairs and to prove that periodic closings of the chain possesses higher symmetries.

The Lax pair for equation \eqref{eq1}
\begin{equation*} %\label{Lax1eps}
\psi_{n, x} = \frac{u_{n, x}}{u_{n+1} - u_n} (\psi_{n+1} - \psi_n), \qquad \psi_{n, y} = \frac{u_{n, y}}{u_{n} - u_{n-1}} (\psi_{n} - \psi_{n-1})
\end{equation*}
was found by E.V.~Ferapontov. To construct Lax pair for chain \eqref{eq3}, we use the scheme suggested in paper~\cite{FerapNovR}. Let us describe the procedure in detailed. First of all, we represent lattice (\ref{eq3}) in the equivalent following form:
\begin{equation}
u_{xy}=\big(u_x-u^2-1\big)\big(u_y-u^2-1\big)\frac{\triangle_{z\bar z}u}{\triangle_z u \triangle_{\bar z}u}+2u\big(u_x+u_y-u^2-1\big).
\label{H2}
\end{equation}
Here $\triangle_{ z}=\frac{T_z-1}{\epsilon}$, $\triangle_{\bar z}=\frac{1-T_{\bar z}}{\epsilon}$ are the forward/backward discrete derivatives and $\triangle_{z \bar z}=\frac{T_z+T_{\bar z}-2}{\epsilon^2}$ is the symmetrised second-order discrete derivative; the operators $T_z$, $T_{\bar z}$ are the forward and backward $\epsilon$-shifts operators in the variable $z$.

The method consists of three steps:
\begin{enumerate}\itemsep=0pt
\item[1)] First we construct the dispersionless limit of the equation (obtained as $\epsilon \rightarrow 0$).

\item[2)] Secondly, for the equation found at the previous step we find dispersionless Lax pair. Usually this problem is effectively solved.

\item[3)] Finally, we reconstruct Lax pair by appropriate ``quantization'' of dispersionless Lax pair as proposed in~\cite{Zakharov}.
\end{enumerate}

The paper is organized as follows. In Section~\ref{section2} Lax pair for chain \eqref{eq3} is constructed. Section~\ref{section3} is devoted to periodic closings. Namely, we impose the periodic closure conditions $u_{n+2} = u_n$ to infinite chains~\eqref{eq1},~\eqref{eq3} and obtain finite systems. Lax pairs and higher symmetries of the second order are constructed for obtained finite systems. Conclusion contains a~discussion of the results.

\section{Construction Lax pair for equation (\ref{eq3})}\label{section2}

The main result of this section is as follows:
\begin{Theorem}\label{theorem2.1}
Equation \eqref{eq3} possesses the Lax pair
\begin{gather*}
\psi_{n, x} = \frac{u_{n,x} - u^2_n-1}{u_{n+1} - u_n} (\psi_{n+1} - \psi_n) + u_n \psi_n,\\
\psi_{n, y} = \frac{u_{n, y} - u^2_n-1}{u_n - u_{n-1}} (\psi_n - \psi_{n-1}) + u_n \psi_n.
\end{gather*}
\end{Theorem}

\begin{proof}
The dispersionless limit of the equation (\ref{H2}) coincides with equation:
\begin{equation} \label{H2disp}
u_{xy}=\big(u_x-u^2-1\big)\big(u_y-u^2-1\big)\frac{u_{zz}}{u^2_z}+2u\big(u_x+u_y-u^2-1\big).
\end{equation}
There exists a direct method for finding Lax pairs for equations of this form. Lax pair is sought in the following form:
\begin{gather}
 S_x = F (u, u_x, u_y, u_z, S_z), \label{Lax2_eq1}\\
 S_y = G(u, u_x, u_y, u_z, S_z).	\label{Lax2_eq2}
\end{gather}
The compatibility condition $S_{xy} = S_{yx}$ of system \eqref{Lax2_eq1}, \eqref{Lax2_eq2} by virtue of equation (\ref{H2disp}) leads to the overdetermined equation
\begin{gather*}
F_{u_y} u_{yy} u^2_z\! - G_{u_x} u_{xx} u^2_z\!- ( G_{S_z}F_{u_x} - G_{u_x} F_{S_z}+ G_{u_z} )u_{zx}u^2_z\!
- (G_{u_y} F_{S_z} - G_{S_z}F_{u_x} - F_{u_z} ) u_{zy} u^2_z \nonumber\\
\qquad{} +u_{zz} \big( \big(u^2 -u_y + 1\big)\big(u^2 - u_x + 1\big) (F_{u_x} - G_{u_y} ) -u^2_z ( G_{S_z}F_{u_z} - G_{u_z}F_{S_z} ) \big) \nonumber\\
\qquad{} - u^2_z \big( 2u\big(1 + u^2 -u_x -u_y\big) ( F_{u_x} - G_{u_y} ) + u_z ( G_{S_z}F_{u} - G_{u}F_{S_z} ) + u_x G_{u} - u_yF_{u}\big) = 0. %\label{comp+cond2}
\end{gather*}
Because of the fact that variables $u$, $u_x$, $u_y$, $u_z$, $u_{xx}$, $u_{yy}$, $u_{zx}$, $u_{zy}$, $u_{zz}$ are independent, this equation splits down into the overdetermined system of equations:
\begin{gather}
F_{u_y} = 0, \qquad G_{u_x} = 0, \label{eqs3_5}\\
G_{S_z}F_{u_x} - G_{u_x} F_{S_z}+ G_{u_z} = 0, \label{eq3_6}\\
G_{u_y} F_{S_z} - G_{S_z}F_{u_x} - F_{u_z} = 0,	\label{eq3_7}\\
\big(u^2 -u_y + 1\big)\big(u^2 - u_x + 1\big) (F_{u_x} - G_{u_y} ) -u^2_z ( G_{S_z}F_{u_z} - G_{u_z}F_{S_z} ) = 0,	\label{eq3_8}\\
 2u\big(1 + u^2 -u_x -u_y\big) ( F_{u_x} - G_{u_y} )
+ u_z ( G_{S_z}F_{u} - G_{u}F_{S_z} ) + u_x G_{u} - u_yF_{u} = 0.	\label{eq3_9}
\end{gather}
Equations (\ref{eqs3_5}) mean that $F = F(u, u_x, u_z, S_z)$ and $G = G(u, u_y, u_z, S_z)$. Substituting~$F$ and~$G$ into (\ref{eq3_6}), (\ref{eq3_7}), we arrive at the equations:
\begin{equation}
G_{u_z} + G_{S_z}F_{u_x} = 0, \qquad
F_{u_z} + G_{u_y} F_{S_z} = 0. \label{Eqs3_10}
\end{equation}
We differentiate the first equation (\ref{Eqs3_10}) by $u_x$, the second equation (\ref{Eqs3_10}) -- by $u_y$, and obtain that $G_{S_z}F_{u_x u_x} = 0$, $F_{S_z} G_{u_y u_y} = 0$. Obviously that the functions $F$ and $G$ take the following forms:
\begin{gather*}
 F (u, u_x, u_z, S_z) = F_2 (u, u_z, S_z) u_x + F_3(u, u_z, S_z), \\
 G(u, u_y, u_z, S_z) = F_4(u, u_z, S_z) u_y + F_5 (u, u_z, S_z).
\end{gather*}
Then we rewrite (\ref{Eqs3_10}) and (\ref{eq3_8}), (\ref{eq3_9}) using the last formulas. Because of the fact that the variables $u$, $u_x$, $u_y$, $u_z$ are independent, obtained equations split down one more time. Thus we arrive at the system for unknown functions $F_i(u, u_z, S_z)$, $i=2,3,4,5$:
\begin{gather} \label{eq3_13}
F_2 F_{4,S_z} + F_{4,u_z} = 0, \qquad F_4 F_{2,S_z} + F_{2,u_z} = 0,
\\	
F_2 - F_4 + u^2_z ( F_{2,S_z} F_{4,u_z} - F_{4,S_z}F_{2,u_z} ) = 0, \label{eq3_15} \\
F_{4,u} - F_{2,u} + u_z ( F_{2,u} F_{4,S_z}- F_{4,u} F_{2,S_z} ) = 0, \label{eq3_16}\\
F_4 F_{3,S_z} + F_{3,u_z} = 0, \label{eq3_17}\\
\big(1 + u^2\big)(F_4 - F_2) +u^2_z ( F_{3,S_z} F_{4,u_z} - F_{4,S_z}F_{3,u_z} ) =0, \label{eq3_18}\\
2u(F_4 - F_2) +u_z ( F_{3,u} F_{4,S_z}- F_{4,u} F_{3,S_z} )- F_{3,u} = 0, 	\label{eq3_19}\\
F_2 F_{5,S_z} + F_{5,u_z} = 0, \label{eq3_20}\\
2u( F_4 - F_2) + u_z ( F_{2,u} F_{5,S_z} - F_{5,u} F_{2,S_z} ) + F_{5,u} = 0, \label{eq3_21}\\
\big(u^2 + 1\big) (F_4 - F_2) + u^2_z ( F_{2,S_z} F_{5,u_z} - F_{5,S_z} F_{2,u_z} ) = 0, \label{eq3_22}\\
\big(u^2 + 1\big)^2 (F_2 - F_4) + u^2_z ( F_{3,S_z} F_{5,u_z} - F_{5,S_z} F_{3,u_z} ) = 0, \label{eq3_23}\\
2u\big(u^2+1\big) (F_2 - F_4) + u_z ( F_{3,S_z} F_{5,u} - F_{5,S_z} F_{3,u} ) = 0. \label{eq3_24}
\end{gather}
Now we will work with equations (\ref{eq3_13})--(\ref{eq3_16}) to clarify functions $F_2$, $F_4$. Let us express $F_{4,u_z}$, $F_{2,u_z}$ from (\ref{eq3_13}) and substitute them into (\ref{eq3_16}). This leads to the condition $F_4 = F_2$ or to the equation
\begin{equation} \label{eq3_25}
\big( 1 - u^2_z F_{2,S_z}F_{4,S_z}\big) = 0.
\end{equation}
Let us consider case (\ref{eq3_25}). We look for $F_2$, $F_4$ in the following form:
\begin{equation*} %\label{sec3_F2F4}
F_2 (u, u_z, S_z) = \frac{A(u, S_z)}{u_z}, \qquad F_4 (u, u_z, S_z) = \frac{B(u, S_z)}{u_z}.
\end{equation*}
Then $A$, $B$ have to satisfy the system obtained using (\ref{eq3_25}), (\ref{eq3_13}), and (\ref{eq3_15}),
\begin{gather}
1 - A_{S_z} B_{S_z} = 0, \qquad -A + B A_{S_z} = 0, \qquad -B + A B_{S_z} = 0, \label{eq3_27}\\
B_{u} - A_{u}+ B_{S_z} A_{u} - A_{S_z} B_{u} = 0. \label{eq3_30}
\end{gather}
This system has the solution:
\begin{equation*}
A (u, S_z) =\frac{ {\rm e}^{a_1(u)S_z + a_1(u) a_2(u)} - 1}{a_1(u)}.
\end{equation*}
Here $a_1$, $a_2$ are arbitrary functions. Similarly, we find that
\begin{equation*}
B (u, S_z) =\frac{ {\rm e}^{a_4(u)S_z + a_4(u) a_3(u)} - 1}{a_4(u)}
\end{equation*}
with arbitrary functions $a_3$, $a_4$. Under obtained $A$ and $B$ the first equation (\ref{eq3_27}) becomes
\begin{equation*}
1 - {\rm e}^{(a_1(u)+a_4(u))S_z + a_2(u)a_1(u)+a_3(u)a_4(u)} = 0.
\end{equation*}
Thus one can derive that $a_4 = -a_1$, $a_3 = a_2$. Finally, equation (\ref{eq3_30}) takes the form
\begin{gather*}
\big( {-} a_1(u) a'_1(u) S_z - a^2_1(u) a'_2(u) - a_1(u) a_2(u) a'_1(u) + 2 a'_1(u) \big) {\rm e}^{a_1(u)(S_z+a_2(u))} \\
{}+\big( a^2_1(u)a'_2(u) + a_1(u) a_2(u) a'_1(u) + 2 a'_1(u) + a_1(u)a'_1(u)S_z \big) {\rm e}^{-a_1(u)(S_z+a_2(u))} - 4 a'_1(u) = 0.
\end{gather*}
We assume essential dependence on $S_z$ for functions $F_2$, $F_4$ and, therefore, for $A$, $B$, so the functions ${\rm e}^{a_1(u)S_z}$, ${\rm e}^{-a_1(u)S_z}$, ${\rm e}^{a_1(u)S_z}S_z$, ${\rm e}^{-a_1(u)S_z}S_z$ are independent. Hence we have $a_1(u) = c_1$, $a_2(u) = c_2$, where $c_1$, $c_2$ are arbitrary constants.

Thus, we have clarified the right hand sides of Lax pair (\ref{Lax2_eq1}), (\ref{Lax2_eq2})
\begin{gather*}
S_x = F(u, u_x, u_y, u_z, S_z) = \frac{\big({\rm e}^{c_1 (S_z + c_2)} - 1\big)u_x}{c_1 u_z} + F_3(u, u_z, S_z),\\
S_y = G(u, u_x, u_y, u_z, S_z) = -\frac{\big({\rm e}^{-c_1 (S_z + c_2)} - 1\big)u_y}{c_1 u_z} + F_5(u, u_z, S_z).
\end{gather*}
By the shift transformation $S \rightarrow S - c_2 z$ and by the scaling $z \rightarrow c_1 z$ these equations can be reduced to
\begin{gather*}
S_x = F(u, u_x, u_y, u_z, S_z) = \frac{\big({\rm e}^{ S_z } - 1\big)u_x}{u_z} + F_3(u, u_z, S_z),\\
S_y = G(u, u_x, u_y, u_z, S_z) = -\frac{\big({\rm e}^{-S_z } - 1\big)u_y}{ u_z} + F_5(u, u_z, S_z).
\end{gather*}
To clarify $F_3$, we substitute the above functions into (\ref{eq3_17}), (\ref{eq3_18}), and (\ref{eq3_19})
\begin{gather*}
\big({\rm e}^{-S_z} - 1\big) u_z F_{3,u} - 2u\big( {\rm e}^{S_z} + 2\big) = 0,
\\
 u_z \big({\rm e}^{-S_z} - 1 - u_z {\rm e}^{-S_z}\big)F_{3,S_z} - \big(u^2 + 1\big) \big({\rm e}^{S_z} + {\rm e}^{-S_z} - 2\big) = 0,
\\
-({\rm e}^{-S_z} - 1) F_{3,S_z} + u_z F_{3,u_z} = 0.
\end{gather*}
This system has the solution
\[
F_3(u, u_z, S_z) = -\frac{\big({\rm e}^{S_z} - 1\big)\big(u^2+1\big)}{u_z}.
\] Now we rewrite equations (\ref{eq3_20})--(\ref{eq3_24}) and we obtain the system on the unknown function $F_5$:
\begin{gather*}
 \big({\rm e}^{S_z} - 1\big) F_{5,S_z} + u_z F_{5,u_z} = 0,\\
 u_z \big({\rm e}^{S_z} - 1\big) F_{5,S_z} + u^2_z {\rm e}^{S_z} F_{5,u_z} - \big(u^2+1\big)\big({\rm e}^{S_z} + {\rm e}^{-S_z} - 2\big) = 0,\\
 -u_z(-{\rm e}^{-2 S_z} + 3 {\rm e}^{-S_z} - 3 + {\rm e}^{S_z}) F_{5,S_z} - u^2_z({\rm e}^{S_z} + {\rm e}^{-S_z} - 2) F_{5,u_z} \nonumber \\
 \qquad{} + \big(u^2+1\big) \big({-}4 + {\rm e}^{S_z} - 4 {\rm e}^{-2 S_z} + 6 {\rm e}^{-S_z} + {\rm e}^{-3 S_z}\big) = 0,\\
 u_z\big(1- {\rm e}^{S_z}\big) F_{5,u} -2u\big( {\rm e}^{S_z} + {\rm e}^{-S_z} - 2\big) = 0,\\
 2 u u_z\big({\rm e}^{-2 S_z} - 3 {\rm e}^{-S_z} + 3 - {\rm e}^{S_z}\big) F_{5,S_z} + \big(u^2+1\big) u_z \big({\rm e}^{S_z} + {\rm e}^{-S_z} - 2\big) F_{5,u} \nonumber \\
 \qquad{}+ 2u\big(u^2+1\big)\big({\rm e}^{-3 S_z} + 6 {\rm e}^{-S_z} - 4 {\rm e}^{-2 S_z} + {\rm e}^{S_z} - 4\big) = 0.
\end{gather*}
This system possesses the solution
\[
F_5 (u, u_z, S_z) = -\frac{\big(1 - {\rm e}^{-S_z}\big)\big(u^2+1\big)}{u_z}.
\]

Thus we have found the Lax pair
\begin{gather}
 S_x = \frac{u_x - u^2 - 1}{u_z} \big({\rm e}^{S_z} - 1\big) + \frac{1}{u_z},	\label{Lax2_eq11}\\
 S_y = \frac{u_y - u^2 - 1}{u_z} \big(1 - {\rm e}^{-S_z}\big) - \frac{1}{u_z} \label{Lax2_eq12}
\end{gather}
for equation (\ref{H2disp}).

Now we reconstruct the dispersive Lax pair by an appropriate quantization the dispersionless Lax pair (\ref{Lax2_eq11}), (\ref{Lax2_eq12}). First, we ``quantise'' \cite{Zakharov} the terms in every equation (\ref{Lax2_eq11}), (\ref{Lax2_eq12}): $u_z$~is replaced by $\triangle_z u$; ${\rm e}^{S_z} - 1$ by $\triangle_z \psi$ due to the formal representation ${\rm e}^{\frac{\partial}{\partial z}} \approx 1 + \frac{\partial}{\partial z} + \cdots$, and, similarly $1 - {\rm e}^{-S_z}$ by $\triangle_{\bar z} \psi$.

In most cases, this procedure provides the necessary Lax pair. But in this case we do not obtain the Lax pair for (\ref{H2}) if we act in the same way. It was experimentally found that we should fit the second term in the r.h.s.\ of equations (\ref{Lax2_eq11}), (\ref{Lax2_eq12}) by the following way (i.e., we guess some part):
\begin{gather*}
 \psi_x = \frac{u_x - u^2 - 1}{\triangle_z u} \triangle_z \psi + P(u) \psi,	\\
 \psi_y = \frac{u_y - u^2 - 1}{\triangle_{\bar z} u} \triangle_{\bar z} \psi + Q(u) \psi.
\end{gather*}
The compatibility condition $\psi_{xy} = \psi_{yx}$ is straightforward to solve. Thus we find that
 equation~(\ref{H2}) possesses the Lax pair
\begin{gather*}
 \psi_x = \frac{u_x - u^2 - 1}{\triangle_z u} \triangle_z \psi + u \psi,\qquad
 \psi_y = \frac{u_y - u^2 - 1}{\triangle_{\bar z} u} \triangle_{\bar z} \psi + u \psi.
\end{gather*}
It finally proved Theorem~\ref{theorem2.1}. \end{proof}

\section{Higher symmetries of periodic closings}\label{section3}

Let us impose the periodic closure conditions $u_{n+2} = u_n$ to infinite lattice \eqref{eq1}. Then we obtain the following finite system:
\begin{equation}
u_{0,xy} = \frac{2}{u_0 - u_1} u_{0,x} u_{0,y}, \qquad u_{1,xy} = \frac{2}{u_1 - u_0} u_{1,x} u_{1,y}. \label{eq31}
\end{equation}
System \eqref{eq31} has the $x$-integral and the $y$-integral
\begin{equation} \label{Int1}
w = \frac{u_{0,y} u_{1,y}}{(u_0 - u_1)^2}, \qquad W = \frac{u_{0,x} u_{1,x}}{(u_0 - u_1)^2}.
\end{equation}
Lax pair for \eqref{eq31} has the form
\begin{equation} \label{Lax1}
\Psi_x = (A \lambda + B) \Psi, \qquad \Psi_y = \big(\tilde{A} \lambda^{-1}+ \tilde{B}\big) \Psi,
\end{equation}
where $\Psi = (\psi_1, \psi_0)^{\rm T}$ and
\begin{gather*}
A = \left(\begin{matrix}
0 & 0 \\
\dfrac{u_{1,x}}{u_0 - u_1} & 0
\end{matrix} \right), \qquad
B = \left(\begin{matrix}
\dfrac{u_{0,x}}{u_0 - u_1} & -\dfrac{u_{0,x}}{u_0 - u_1}\vspace{1mm}\\
0 & -\dfrac{u_{1,x}}{u_0 - u_1}
\end{matrix} \right),
\\
\tilde{A} = \left(\begin{matrix}
0 & 0 \\
0 & -\dfrac{u_{0,y}}{u_0 - u_1}
\end{matrix} \right), \qquad
\tilde{B} = \left(\begin{matrix}
-\dfrac{u_{0,y}}{u_0 - u_1} & 0\vspace{1mm}\\
-\dfrac{u_{1,y}}{u_1 - u_0} & -\dfrac{u_{1,y}}{u_1 - u_0}
\end{matrix} \right),
\end{gather*}
$\lambda$ is a spectral parameter.

The classical symmetry can be found directly from the consistency condition $(u_{i,xy})_{t_1} = (u_{i,t_1})_{xy}$:
\begin{gather*}
u_{0,t_1} = u_{0,x} F(W) + c_1 u^2_0 + c_2 u_0 + c_3,\\
u_{1,t_1} = u_{1,x} F(W) + c_1 u^2_1 + c_2 u_1 + c_3,
\end{gather*}
where $F$ is an arbitrary function depending on the $y$-integral $W$ defined by the second formula of \eqref{Int1}; $c_1$, $c_2$, $c_3$ are arbitrary constants. The classical symmetry in the another direction is simply found because the system is symmetric under the change of variables $x \leftrightarrow y$:
\begin{gather*}
 u_{0,t_2} = u_{0,y} G(w) + \tilde{c}_1 u^2_0 + \tilde{c}_2 u_0 + \tilde{c}_3,\\
 u_{1,t_2} = u_{1,y} G(w) + \tilde{c}_1 u^2_1 + \tilde{c}_2 u_1 + \tilde{c}_3.
\end{gather*}
Higher symmetry of the second order is sought in the following form:
\begin{gather*}
u_{i,\tau_1} = a_i(u_0,u_1,u_{0,x},u_{1,x})u_{0,xx}+b_i(u_0,u_1,u_{0,x},u_{1,x})u_{1,xx}+h_i(u_0,u_1,u_{0,x},u_{1,x}),
\end{gather*}
$i=1,2$, where $a_i$, $b_i$, $h_i$ are functions to be found. To find the higher symmetry we use Lax pair~\eqref{Lax1}. Let us consider the linear problem
\begin{equation} \label{Psi_t}
\Psi_{\tau_1} = \big(\alpha \lambda^2 + \beta \lambda + \gamma\big) \Psi,
\end{equation}
where $\alpha = (\alpha_{i,j})$, $\beta = (\beta_{i,j})$, $\gamma = (\gamma_{i,j})$, $i,j = 1, 2$ are matrices to be found. It is assumed that elements of the matrices depend on the variables $u_{0}$, $u_1$, $u_{0,x}$, $u_{1,x}$, $u_{0,xx}$, $u_{1,xx}$. The compatibility condition $(\Psi_x)_{\tau_1} = (\Psi_{\tau_1})_x$ for the systems
\begin{equation*}
\Psi_x = (A \lambda + B) \Psi, \qquad \Psi_{\tau_1} = \big(\alpha \lambda^2 + \beta \lambda + \gamma\big) \Psi,
\end{equation*}
results in the system of relations
\begin{gather*}
 A \alpha = \alpha A, \qquad A \beta + B \alpha = \alpha_x + \alpha B + \beta A,\\
 A_{\tau_1} + A \gamma + B \beta = \beta_x + \beta B + \gamma A, \qquad
B_{\tau_1} + B \gamma = \gamma_x + \gamma B.
\end{gather*}
A complete study of these equations leads to the following formulas:
\begin{gather*}
 u_{0,\tau_1} = H(W) u_{0,xx} + \frac{u^2_{0,x}}{(u_0 - u_1)^2} \Phi(W) u_{1,xx} + (u_0 - u_1) g(u_0,u_1,u_{0,x}, u_{1,x}) \\
\hphantom{u_{0,\tau_1} =}{} + (u_0 - u_1) (c_0 - c_1 u_1 - \frac{c_2}{2}) - (c_1 u^2_1 + c_2 u_1 + c_3),\\
 u_{1,\tau_1} = \frac{u_{1,x}}{u_{0,x}} H(W) u_{0,xx} + W \Phi(W) u_{1,xx} + \frac{(u_0 - u_1)u_{1,x}}{u_{0,x}} g(u_0,u_1,u_{0,x}, u_{1,x}) \\
\hphantom{u_{1,\tau_1} =}{} + \frac{(u_0 - u_1)u_{1,x}}{u_{0,x}} (c_0 + c_1 u_0 + \frac{c_2}{2}) - (c_1 u^2_1 + c_2 u_1 + c_3),
\end{gather*}
where $H$, $\Phi$, $g$ are arbitrary functions; $c_i$ are arbitrary constants. To define precisely obtained formulas we substitute them into the compatibility condition $(u_{i,xy})_{\tau_1} = (u_{i,\tau_1})_{xy}$. Thus, we finally found the higher symmetry of the second order:
\begin{gather} \label{u0t_1}
u_{0,\tau_1} = \left( u_{0,xx} + \frac{u_{0,x}}{u_{1,x}} u_{1,xx} -\frac{2 u_{0,x} (u_{0,x} - u_{1,x})}{(u_0 - u_1)} \right) F(W),
\\
	\label{u1t_1}
u_{1,\tau_1} = \left( u_{1,xx} + \frac{u_{1,x}}{u_{0,x}} u_{0,xx} - \frac{2 u_{1,x} (u_{0,x} - u_{1,x})}{(u_0 - u_1)} \right) F(W),
\end{gather}
where $F$ is an arbitrary function; $W$ is the $y$-integral defined by the second formula of \eqref{Int1}. Also we finally found matrices $\alpha$, $\beta$, $\gamma$ involved in \eqref{Psi_t}:
\begin{gather*}
\alpha = \left(\begin{matrix}
\alpha_{11} & 0\\
0 & \alpha_{11}
\end{matrix}\right), \qquad
\beta = \left( \begin{matrix}
\beta_{11} & 0\\
\beta_{21}(\overline{u}, \overline{u}_x, \overline{u}_{xx}) & \beta_{11}
\end{matrix}\right),
\\
\gamma = \left(\begin{matrix}
\gamma_{11}(\overline{u}, \overline{u}_x, \overline{u}_{xx}) & \gamma_{12}(\overline{u}, \overline{u}_x, \overline{u}_{xx})\\
0 & \gamma_{22}(\overline{u}, \overline{u}_x, \overline{u}_{xx})
\end{matrix}\right),
\end{gather*}
where
\begin{gather*}
\beta_{21}(\overline{u}, \overline{u}_x, \overline{u}_{xx}) =\left( \frac{u_{1,x}}{u_{0,x} (u_0 - u_1)} u_{0,xx} + \frac{1}{(u_0 - u_1)} u_{1,xx} -
 \frac{2 u_{1,x} (u_{0,x}-u_{1,x})}{(u_0 - u_1)^2}\right) F(W),
\\
\gamma_{11}(\overline{u}, \overline{u}_x, \overline{u}_{xx}) = \left( \frac{1}{(u_0 - u_1)} u_{0,xx} + \frac{u_{0,x}}{u_{1,x}(u_0 - u_1)} u_{1,xx}
-\frac{2 u_{0,x}(u_{0,x} - u_{1,x})}{(u_0 - u_1)^2} \right) F(W),
\\
\gamma_{12}(\overline{u}, \overline{u}_x, \overline{u}_{xx}) = \left( -\frac{1}{(u_0 - u_1)} u_{0,xx} - \frac{u_{0,x}}{u_{1,x} (u_0 - u_1)} u_{1,xx}
+ \frac{u_{0,x} (u_{0,x} - u_{1,x})}{(u_0 - u_1)^2} \right) F(W),
\end{gather*}
$\alpha_{11}$, $\beta_{11}$ are arbitrary constants. Thus it is seen that definitive answer is given by formulas~\eqref{u0t_1},~\eqref{u1t_1} and
\begin{equation*}
\Psi_{\tau_1} = (\beta \lambda + \gamma) \Psi, \qquad \beta = \left( \begin{matrix}
0 & 0\\
\beta_{21} & 0
\end{matrix}\right), \qquad \gamma = \left(\begin{matrix}
\gamma_{11} & \gamma_{12}\\
0 & \gamma_{22}
\end{matrix}\right),
\end{equation*}
where $\beta_{21}$, $\gamma_{ij}$ have been described just above.

\begin{Remark}
The symmetry given by \eqref{u0t_1}, \eqref{u1t_1} can be written as\footnote{I am grateful to the anonymous referee for this constructive comment.}
\begin{equation*}
u_{0,\tau_1}=u_{0,x} F(W) \frac{W_x}{W}, \qquad
u_{1,\tau_1}=u_{1,x} F(W) \frac{W_x}{W}.
\end{equation*}
Therefore this is actually the classical symmetry in disguise.
\end{Remark}

Let us consider chain \eqref{eq3}. We impose the periodic closure conditions $u_{n+2} = u_n$ to infinite chain \eqref{eq3} and obtain the following finite system:
\begin{gather}
 u_{0,xy} = \frac{2}{u_0 - u_1} \big(u_{0,x} - u^2_0 - 1\big)\big(u_{0,y} - u^2_0 - 1\big) + 2 u_0 \big(u_{0,x} + u_{0,y} - u^2_0 - 1\big), \nonumber\\
 u_{1,xy} = \frac{2}{u_1 - u_0} \big(u_{1,x} - u^2_1 - 1\big)\big(u_{1,y} - u^2_1 - 1\big) + 2 u_1 \big(u_{1,x} + u_{1,y} - u^2_1 - 1\big). 	\label{fin_sys2}
\end{gather}
This system possesses the $y$-integral and the $x$-integral
\begin{equation} \label{Int2}
P = \frac{\big(u_{0,x} - u^2_0 - 1\big)\big(u_{1,x} - u^2_1 - 1\big)}{(u_0 - u_1)^2}, \qquad J = \frac{\big(u_{0,y} - u^2_0 - 1\big)\big(u_{1,y} - u^2_1 - 1\big)}{\big(u_0 - u_1\big)^2}.
\end{equation}
System \eqref{fin_sys2} is the compatibility condition for the Lax pair
\begin{equation}	\label{Lax2}
\Phi_x = (S \lambda + T) \Phi, \qquad \Phi_y = \big(\tilde{S} \lambda^{-1} + \tilde{T}\big)\Phi,
\end{equation}
where $\Phi = (\phi_0, \phi_1)^{\rm T}$,
\begin{gather*}
S = \left(\begin{matrix}
0 & 0\\
\dfrac{u_{1,x} - u^2_1 - 1}{u_0 - u_1} & 0
\end{matrix}\right), \qquad
T = \left(\begin{matrix}
- \dfrac{u_{0,x} - u^2_0 - 1}{u_1 - u_0} + u_0 & \dfrac{u_{0,x} - u^2_0 - 1}{u_1 - u_0} \vspace{1mm}\\
0 & -\dfrac{u_{1,x} - u^2_1 - 1}{u_0 - u_1} + u_1
\end{matrix}\right),
\\
\tilde{S} = \left( \begin{matrix}
0 & -\dfrac{u_{0,y} - u^2_0 - 1}{u_0 - u_1}\\
0 & 0
\end{matrix} \right), \qquad
\tilde{T} = \left( \begin{matrix}
\dfrac{u_{0,y} - u^2_0 - 1}{u_0 - u_1} + u_0 & 0\vspace{1mm}\\
-\dfrac{u_{1,y} - u^2_1 - 1}{u_1 - u_0} & \dfrac{u_{1,y} - u^2_1 - 1}{u_1 - u_0} + u_1
\end{matrix} \right).
\end{gather*}
To find the higher symmetry it is sufficient (as we have just seen) to consider the system
\begin{equation} \label{Psi_tau2}
\Phi_{\tau_2} = \big(\tilde{\beta} \lambda + \tilde{\gamma}\big) \Phi,
\end{equation}
compatible with the first equation of \eqref{Lax2}.
In this way we obtained the higher symmetry of system~\eqref{fin_sys2}:
\begin{gather}	\label{sym21}
u_{0,\tau_2} = \left( u_{0,xx} + \frac{u_{0,x} - u^2_0 - 1}{u_{1,x} - u^2_1 - 1} u_{1,xx}
+ \frac{ 2 \varphi(u_0, u_1, u_{0,x}, u_{1,x})}{(u_{1,x} - u^2_1 - 1) (u_0 - u_1)} \right) F(P) ,
\\ \label{sym22}
u_{1,\tau_2} = \left(\frac{ u_{1,x} - u^2_1 - 1}{u_{0,x} - u^2_0 - 1} u_{0,xx} + u_{1,xx}
+ \frac{2 \varphi(u_0, u_1, u_{0,x}, u_{1,x}) }{(u_{0,x} - u^2_0 - 1)(u_0 - u_1)} \right)F(P),
\end{gather}
where $P$ is the $y$-integral given by the first formula of \eqref{Int2},
\begin{gather}
\varphi(u_0,u_1,u_{0,x},u_{1,x}) = u_{0,x} u_{1,x} (u_{1,x} - u_{0,x}) + u^2_{0,x} \big(1 + u^2_1\big) - u^2_{1,x} \big(1 + u^2_0\big) \nonumber \\
\hphantom{\varphi(u_0,u_1,u_{0,x},u_{1,x}) =}{} - u_{0,x} \big(1 + u^2_1 + u_0 u_1 + u_0 u^3_1\big) + u_{1,x} \big(1 + u^2_0 + u_0 u_1 + u^3_0 u_1\big).\label{phi2}
\end{gather}
Matrices $\tilde{\beta}$, $\tilde{\gamma}$ (see \eqref{Psi_tau2}) are defined by the following formulas:
\begin{gather*}
\tilde{\beta} = \left( \begin{matrix}
0 & 0\\
\tilde{\beta}_{21}(\bar{u}, \bar{u}_x, \bar{u}_{xx}) & 0
\end{matrix}
\right), \qquad
\tilde{\gamma} = \left(\begin{matrix}
\tilde{\gamma}_{11}(\bar{u}, \bar{u}_x, \bar{u}_{xx}) & \tilde{\gamma}_{12}(\bar{u}, \bar{u}_x, \bar{u}_{xx}) \\
0 & \tilde{\gamma}_{22}(\bar{u}, \bar{u}_x, \bar{u}_{xx})
\end{matrix}\right),
\end{gather*}
where
\begin{gather*}
\tilde{\beta}_{21}(\bar{u}, \bar{u}_x, \bar{u}_{xx}) = \left( \frac{u_{1,x} - u^2_1 - 1}{(u_0 - u_1)\big(u_{0,x} - u^2_0 - 1\big)} u_{0,xx} + \frac{u_{1,xx}}{u_0 - u_1} \right.\\
\left. \hphantom{\tilde{\beta}_{21}(\bar{u}, \bar{u}_x, \bar{u}_{xx}) =}{} + \frac{2 \varphi(u_0, u_1, u_{0,x}, u_{1,x}) }{\big(u_{0,x} - u^2_0 - 1\big) (u_0 - u_1)^2}\right) F(P),
\\
\tilde{\gamma}_{11}(\bar{u}, \bar{u}_x, \bar{u}_{xx}) = \left(\frac{ u_{0,xx}}{u_0 - u_1} + \frac{\big(u_{0,x} - u^2_0 - 1\big)u_{1,xx}}{\big(u_{1,x} - u^2_1 - 1\big)(u_0 - u_1)} \right.\\
\left. \hphantom{\tilde{\gamma}_{11}(\bar{u}, \bar{u}_x, \bar{u}_{xx}) =}{}
+ \frac{2\varphi(u_0, u_1, u_{0,x}, u_{1,x})}{(u_0 - u_1)^2\big(u_{1,x} - u^2_1 - 1\big)}\right) F(P) ,
\\
\tilde{\gamma}_{12}(\bar{u}, \bar{u}_x, \bar{u}_{xx}) = \left(-\frac{u_{0,xx}}{u_0 - u_1} - \frac{\big(u_{0,x} - u^2_0 - 1\big)}{\big(u_{1,x} - u^2_1 - 1\big)(u_0 - u_1)}u_{1,xx} \right.\\
\left. \hphantom{\tilde{\gamma}_{12}(\bar{u}, \bar{u}_x, \bar{u}_{xx}) =}{} - \frac{2\varphi(u_0, u_1, u_{0,x}, u_{1,x}) }{(u_0 - u_1)^2 \big(u_{1,x} - u^2_1 - 1\big)}\right)F(P),
\\
\tilde{\gamma}_{22}(\bar{u}, \bar{u}_x, \bar{u}_{xx}) = \left( -\frac{\big(u_{1,x} - u^2_1 - 1\big)}{(u_0 - u_1)\big(u_{0,x} - u^2_0 - 1\big)} u_{0,xx} - \frac{ u_{1,xx}}{u_0 - u_1} \right.\\
\left. \hphantom{\tilde{\gamma}_{22}(\bar{u}, \bar{u}_x, \bar{u}_{xx}) =}{}
- \frac{2 \varphi(u_0, u_1, u_{0,x}, u_{1,x})}{\big(u_{0,x} - u^2_0 - 1\big)(u_0 - u_1)^2} \right) F(P),
\end{gather*}
$\varphi(u_0, u_1, u_{0,x}, u_{1,x})$ is defined by \eqref{phi2}.

\begin{Remark}
The symmetry given by \eqref{sym21}, \eqref{sym22} can be written as
\begin{equation*}
u_{0,\tau_2}=\big(u_{0,x} - u^2_0 - 1\big) F(P) \frac{P_x}{P}, \qquad
u_{1,\tau_2}=\big(u_{1,x} - u^2_1 - 1\big) F(P) \frac{P_x}{P}.
\end{equation*}
Therefore this is actually the classical symmetry in disguise.
\end{Remark}

Note, that periodic closing obtained by the conditions $u_{n+3} = u_n$ imposing on infinite chain~\eqref{eq3} leads to the system
\begin{gather*}
u_{0,xy} = \left( \frac{1}{u_0 - u_2} - \frac{1}{u_1 - u_0} \right) \big(u_{0,x} - u^2_0 - 1\big)\big(u_{0,y} - u^2_0 - 1\big)+ 2u_0\big(u_{0,x} + u_{0,y} - u^2_0 - 1\big),\\
u_{1,xy} = \left( \frac{1}{u_1 - u_0} - \frac{1}{u_2 - u_1} \right) \big(u_{1,x} - u^2_1 - 1\big)\big(u_{1,y} - u^2_1 - 1\big)+ 2u_1\big(u_{1,x} + u_{1,y} - u^2_1 - 1\big),\\
u_{2,xy} = \left( \frac{1}{u_2 - u_1} - \frac{1}{u_0 - u_2} \right) \big(u_{2,x} - u^2_2 - 1\big)\big(u_{2,y} - u^2_2 - 1\big)+ 2u_2\big(u_{2,x} + u_{2,y} - u^2_2 - 1\big).
\end{gather*}
This system has $y$-integral and $x$-integral
\begin{gather*}
 W = \frac{\big(u_{0,x} - u^2_0 - 1\big)\big(u_{1,x} - u^2_1 - 1\big)\big(u_{2,x} - u^2_2 - 1\big)}{(u_2 - u_1)(u_0 - u_1)(u_0 - u_2)}, \\
 w = \frac{\big(u_{0,y} - u^2_0 - 1\big)\big(u_{1,y} - u^2_1 - 1\big)\big(u_{2,y} - u^2_2 - 1\big)}{(u_2 - u_1)(u_0 - u_1)(u_0 - u_2)}.
\end{gather*}
Lax pair has the following form:
\begin{equation*}
\Psi_x = (A \lambda + B)\Psi, \qquad \Psi_y = \big(\tilde{A} \lambda^{-1} + \tilde{B}\big) \Psi,
\end{equation*}
where $\Psi = (\psi_0, \psi_1, \psi_2)^{\rm T}$,
\begin{gather*}
A = \left(\begin{matrix}
0 & 0 & 0\\
0 & 0 & 0\\
\dfrac{u_{2,x} - u^2_0 - 1}{u_1 - u_0} & 0 & 0
\end{matrix}\right),\\
B = \left(\begin{matrix}
-\dfrac{u_{0,x} - u^2_0 - 1}{u_1 - u_0} + u_0 & \dfrac{u_{0,x} - u^2_0 - 1}{u_1 - u_0} & 0\vspace{1mm}\\
0 & -\dfrac{u_{1,x} - u^2_1 - 1}{u_2 - u_1} + u_1 & \dfrac{u_{1,x} - u^2_1 - 1}{u_2 - u_1}\vspace{1mm}\\
0 & 0 & -\dfrac{u_{2,x} - u^2_2 - 1}{u_0 - u_2} + u_2
\end{matrix}\right),
\\
\tilde{A} = \left(\begin{matrix}
0 & 0 & - \dfrac{u_{0,y} - u^2_0 - 1}{u_0 - u_2}\\
0 & 0 & 0\\
0 & 0 & 0
\end{matrix} \right),\\
\tilde{B} = \left( \begin{matrix}
\dfrac{u_{0,y} - u^2_0 - 1}{u_0 - u_2} + u_0 & 0 & 0\vspace{1mm}\\
-\dfrac{u_{1,y} - u^2_1 - 1}{u_1 - u_0} & \dfrac{u_{1,y} - u^2_1 - 1}{u_1 - u_0} + u_1 & 0\vspace{1mm}\\
0 & -\dfrac{u_{2,y} - u^2_2 - 1}{u_2 - u_1} & \dfrac{u_{2,y} - u^2_2 - 1}{u_2 - u_1} + u_2
\end{matrix} \right).
\end{gather*}

\section{Conclusion}

The problem of classification multidimensional equations is actively studied by many authors, using different algebraic and geometry approaches \cite{B4, B19, B17, B2, B1, B3,FerKhus1,B15,B11}. We note that the classification algorithm for integrable two-dimensional lattices proposed in our previous papers does not provide any algorithm for constructing the Lax pair.

It is known that finite systems obtained from infinite integrable chains by degenerate boundary conditions imposing at the two points of the form $u_{n+k} = c_1$, $u_{n+s} = c_2$ (where $c_1$, $c_2$ are constants) are integrable in the sense of Darboux (they have complete set of integrals in both characteristic directions, i.e., the number of independent integrals is equal to the order of the system). We study finite systems obtained from infinite chains~\eqref{eq1}, \eqref{eq3} by periodic closure conditions. It is interesting fact that each of these systems also has one $x$-integral and one $y$-integral. We obtained that symmetries of these systems depend on integrals. It is known that Darboux integrable systems possesses symmetries which depend on integrals \cite{SS2008,ZhSS95}. Symmetries of systems with incomplete sets of integrals might depend on these integrals \cite{Kiselev, LSmSh}.
In a~discrete version, this fact is discussed in paper~\cite{Xenitidis}. In papers \cite{St1, St2} an algorithm is proposed which allows one to construct higher symmetries of arbitrary order for some special classes of hyperbolic systems possessing the integrals.

\subsection*{Acknowledgements}

The author gratefully thanks I.T.~Habibulin for assignment the problem and useful discussions, E.V.~Ferapontov for explaining the method of the construction of Lax pairs and S.Ya.~Startsev for valuable comments. The author gratefully thanks anonymous referees for a contribution to improve the paper.

\pdfbookmark[1]{References}{ref}
\LastPageEnding

\end{document}